\documentclass[
  aps,
  prl,
  reprint,
  superscriptaddress,
  longbibliography,
  floatfix
]{revtex4-2}

\usepackage{amsmath,amssymb,amsthm}
\usepackage{mathtools}
\usepackage{tikz}
\usetikzlibrary{angles,arrows.meta,calc,quotes}
\usepackage{bm}
\usepackage[english]{babel}
\usepackage{hyperref}

\newcommand{\Id}{\mathbb{I}}
\DeclareMathOperator{\Tr}{Tr}
\newcommand{\dd}{\mathop{}\!\mathrm{d}}
\newcommand{\E}{\mathbb{E}}

\newtheorem{theorem}{Theorem}

\hypersetup{
  colorlinks = true,
  linkcolor = blue,
  citecolor = blue,
  urlcolor = blue
}

\begin{document}

\title{Bell tests for collider decay processes}

\author{Danilo M. Fucci}
\email[Contact author: ]{danilo.fucci@philosophy.ox.ac.uk}
\affiliation{Faculty of Philosophy, University of Oxford, Oxford OX2 6GG, United Kingdom}

\author{Alexandre C. Orthey Jr.}
\affiliation{Instituto de Física Teórica, Universidade Estadual Paulista, Rua Dr.\ Bento Teobaldo Ferraz, 271, Bloco II, CEP 01140-070 São Paulo, São Paulo, Brazil}
\affiliation{Faculty of Mathematics, Informatics and Mechanics, University of Warsaw, ulica Banacha 2, 02-097 Warsaw, Poland}

\author{Alan J. Barr}
\affiliation{Department of Physics, University of Oxford, Oxford OX1 3PU, United Kingdom}
\affiliation{Merton College, Oxford OX1 4JD, United Kingdom}
\affiliation{Département de Physique Nucl\'eaire et Corpusculaire Universit\'e de Gen\`eve, CH-1211 Gen\`eve, Switzerland}

\author{Christopher G. Timpson}
\affiliation{Faculty of Philosophy, University of Oxford, Oxford OX2 6GG, United Kingdom}

\date{\today}

\begin{abstract}
\noindent We propose Bell tests for collider decays using independently sampled analysis axes as settings and detector regions to assign binary or null outcomes. We show that a Bell violation requires postselection, opening a loophole. However, under measurement independence and operational equivalence of acceptance events across settings, any local model exploiting the loophole must be nonclassical in the sense of measurement contextuality. We derive a Bell inequality violated by singlet correlations when the effective spin-analyzing powers are sufficiently high. Its violation then witnesses Bell nonlocality or measurement contextuality.
\end{abstract}

\maketitle

Bell's theorem showed that assumptions about local hidden variables can be subjected to experimental scrutiny in a framework largely independent of the underlying physical theory~\cite{Bell1964,Brunner2014}. While Bell tests have reached a high level of maturity in atomic, optical, and solid-state platforms~\cite{Rosenfeld2017,Giustina2015,Shalm2015,Hensen2015}, it remains an open problem to formulate an operational Bell test in high-energy collider processes~\cite{Fabbrichesi2021} and to address the associated loopholes~\cite{Barr2024,Fabbrichesi:2025aqp}. The difficulty is that a Bell test is a model-independent input-output experiment in which the measurement settings are external variables. Collider detectors, by contrast, primarily record decay data rather than implement freely chosen measurement settings, with the resulting events interpreted through model-dependent reconstruction procedures~\cite{Abel1992,Timpson2023,Li2024,Low2025}. Recent analyses of entanglement in such systems~\cite{ATLAS2024,CMS2024,Bechtle2025} make it timely to ask whether a genuine collider Bell test can be formulated. This work proposes one such construction.

Measurement-independence and model-independence are both crucial for a Bell test, and in collider processes they become intertwined. For local hidden-variables $\lambda$, let $x$ and $y$ denote each party's measurement settings with corresponding outcomes $a$ and $b$. Measurement independence, also called $\lambda$-independence, is given by $p(x,y|\lambda)=p(x,y)$, or equivalently $p(\lambda|x,y)=p(\lambda)$~\cite{Bell2004,Hall2010}. Together with local causality, it leads to the decomposition
\begin{equation}
  \label{eq:bell-lhv-decomp}
  p(a,b|x,y)=\int \dd\lambda\,\rho(\lambda)p(a|x,\lambda)p(b|y,\lambda).
\end{equation}
In standard Bell test scenarios, this assumption is usually justified by the use of macroscopic measurement devices whose settings can be chosen independently and, ideally, in causal separation from the state-preparation procedure. In collider contexts, by contrast, particle decay is often regarded as a self-measurement process~\cite{Tornqvist:1980af,Hiesmayr:2014jva,Barr:2025avs}. In that case, there is no experimental control over the putative measurement settings: both settings and outcomes are extracted from the same microscopic process whose correlations one wants to test, thereby challenging the $\lambda$-independence assumption. Concurrently, without externally assigned inputs, the observed decay data do not directly define conditional frequencies $p(a,b|x,y)$. They must instead be organized into such a structure through a model of the production and decay process, making model independence difficult to maintain.

Recent studies of Bell nonlocality illustrate this issue. They proceed indirectly using a quantum-field-theoretic model to reconstruct a spin density operator from the measured momenta of the decay products; Bell nonlocality witnesses are then evaluated on the inferred state~\cite{Afik2021,Fabbrichesi2021,Severi2022,Barr2024}. At most, such protocols show that the reconstructed state would violate a Bell inequality if suitable measurements could be performed; by themselves, however, they do not provide the observed conditional probabilities of an actual Bell test.

Here we formulate a collider Bell test architecture that restores the appropriate input-output structure. We introduce measurement settings as independently selected analysis axes $\hat x$ and $\hat y$, which specify, on each wing, pairs of detector regions assigned to the outcomes $+1$ and $-1$. Detections outside the selected regions are assigned the no-click outcome $\varnothing$. Operationally, each setting selects coarse-grained detector regions on each wing. These regions may be calibrated in advance using the known relation between decay directions and spin, but the Bell test data are the resulting detector-level outcomes. We first prove a no-go result: if the full record of where the decay products hit the detector is retained, then the settings only define classical post-processings of the same underlying local POVMs. The resulting statistics therefore admit a local hidden-variable model and cannot violate a Bell inequality, in agreement with Refs.~\cite{Kasday1971,Abel1992}.

Our proposed construction circumvents this no-go result by choosing narrow coarse-graining regions and retaining only runs in which neither wing returns the outcome $\varnothing$, an event we call joint acceptance. This necessarily opens a postselection loophole~\cite{Pearle1970,Larsson2014}, which we analyze using measurement noncontextuality: the requirement that operationally equivalent measurement events---those having identical probabilities for every preparation---be represented identically at the hidden-variable level~\cite{Spekkens2005}. When joint acceptance is operationally equivalent for every pair of settings, this requirement prevents postselection from correlating the hidden variables with the chosen settings. Consequently, any local account in which postselection produces the setting dependence required for an apparent Bell violation must be nonclassical by representing joint acceptance contextually. We refer to this as joint-acceptance contextuality. The requisite operational equivalences are suggested by a spin-analyzer model and can, in principle, be tested experimentally in a theory-neutral way~\cite{Mazurek2021}.

Given the geometry of the problem, we then derive a full-correlation, continuous-input, binary-output Bell inequality. Its violation, together with measurement independence and the operational equivalences just described, rules out every hidden-variable model that is both Bell local and noncontextual with respect to joint acceptance. It therefore witnesses Bell nonlocality or joint-acceptance contextuality. We derive the quantum violation threshold in terms of effective spin-analyzing power and coarse-graining width and, in the End Matter, outline a corresponding experimental protocol.

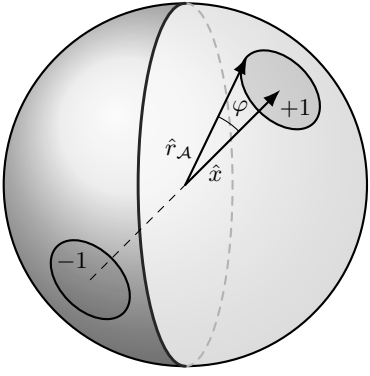
\begin{figure}[t]
  \centering
  \begin{tikzpicture}[scale = 0.8, >=Latex, line cap = round, line join = round]
    \def\sphereRadius{3}
    \def\capAngle{15}
    \def\axisPolarAngle{130}
    \def\axisAzimuth{45}
    \def\meridianRadius{0.78}

    \pgfmathsetmacro{\capRadius}{\sphereRadius*sin(\capAngle)}
    \pgfmathsetmacro{\capYRadius}{\capRadius*abs(cos(\axisPolarAngle))}

    \coordinate (O) at (0,0);
    \coordinate (Xp) at (
      {\sphereRadius*cos(\capAngle)*sin(\axisPolarAngle)*cos(\axisAzimuth)},
      {\sphereRadius*cos(\capAngle)*sin(\axisPolarAngle)*sin(\axisAzimuth)}
    );
    \coordinate (Xm) at (
      {\sphereRadius*cos(\capAngle)*sin(180-\axisPolarAngle)*cos(\axisAzimuth+180)},
      {\sphereRadius*cos(\capAngle)*sin(180-\axisPolarAngle)*sin(\axisAzimuth+180)}
    );

    \shade[ball color = gray!20, opacity = 0.9] (O) circle (\sphereRadius);
    \fill[gray!5, opacity = 0.8]
      (0,\sphereRadius)
      arc[start angle = 90, end angle = -90, radius = \sphereRadius]
      arc[
        start angle = -90,
        end angle = -270,
        x radius = \meridianRadius,
        y radius = \sphereRadius
      ]
      -- cycle;
    \draw[thick] (O) circle (\sphereRadius);
    \draw[dashed, gray!60, line width = 0.8pt]
      (0,-\sphereRadius)
      arc[
        start angle = -90,
        end angle = 90,
        x radius = \meridianRadius,
        y radius = \sphereRadius
      ];
    \draw[black!85, line width = 1.1pt]
      (0,\sphereRadius)
      arc[
        start angle = 90,
        end angle = 270,
        x radius = \meridianRadius,
        y radius = \sphereRadius
      ];

    \draw[dashed, black] (Xm) -- (Xp);
    \draw[->, thick] (O) -- (Xp);
    \node at (0.5,0.2) {$\hat x$};

    \begin{scope}[shift = (Xp), rotate = \axisAzimuth+90]
      \fill[black, opacity = 0.10]
        (0,0) ellipse [x radius = \capRadius, y radius = \capYRadius];
      \draw[thick]
        (0,0) ellipse [x radius = \capRadius, y radius = \capYRadius];
      \coordinate (rp) at (\capRadius,0);
    \end{scope}

    \begin{scope}[shift = (Xm), rotate = \axisAzimuth+270]
      \fill[black, opacity = 0.10]
        (0,0) ellipse [x radius = \capRadius, y radius = \capYRadius];
      \draw[thick]
        (0,0) ellipse [x radius = \capRadius, y radius = \capYRadius];
    \end{scope}

    \draw[->, thick] (O) -- (rp);
    \node at (-0.1,0.6) {$\hat r_{\mathcal A}$};
    \pic[
      draw,
      angle radius = 10mm,
      angle eccentricity = 1.25,
      "$\varphi$"
    ] {angle = Xp--O--rp};

    \node at ($(Xp)+(0.25,-0.3)$) {$+1$};
    \node at ($(Xm)+(-0.3,0.3)$) {$-1$};
  \end{tikzpicture}
  \caption{Geometric coarse graining defined by the analysis axis $\hat x$. The shaded antipodal caps of angular radius $\varphi$ give outcomes $+1$ and $-1$; directions outside them give the geometric no-click outcome $\varnothing$. The vector $\hat r_{\mathcal A}$ is drawn on the boundary of the $+1$ cap, so that $\varphi$ is the angle between $\hat r_{\mathcal A}$ and $\hat x$. The cutaway is only illustrative.}
  \label{fig:caps3d}
\end{figure}

\textit{Model setup.} We consider a source producing a pair of unstable qubits, which we call the parent particles and whose joint state is denoted by $\rho_{\mathcal{AB}}$. Each parent decays into daughter particles, whose detected momenta carry information about the parent spin and therefore act as spin analyzers. We use the standard Bell terminology of Alice and Bob, or of two wings, to refer to the two parent particles and their corresponding decay branches.

On Alice's wing, the relevant decay kinematics are represented by a unit vector $\hat r_{\mathcal A}\in S^2$, corresponding to the daughter momentum direction in the rest frame of the parent particle. The decay induces on the parent qubit an angular POVM density
\begin{equation}
  K_{\mathcal A}(\hat r_{\mathcal A})
  =
  \frac{1}{4\pi}
  \left(
    \Id+\eta_{\mathcal A}\,\hat r_{\mathcal A}\cdot\bm{\sigma}
  \right),
  \quad
  \int_{S^2} \dd\Omega_{\hat r_{\mathcal A}}\,K_{\mathcal A}(\hat r_{\mathcal A})=\Id.
  \label{eq:spin-analyzer-povm}
\end{equation}
Here $\bm{\sigma}=(\sigma_x,\sigma_y,\sigma_z)$ denotes the vector of Pauli matrices acting on the parent qubit, and $\eta_{\mathcal A}\in[0,1]$ is the spin-analyzing power. Operationally, $K_{\mathcal A}(\hat r_{\mathcal A})\,\dd\Omega_{\hat r_{\mathcal A}}$ is the parent-sector effect associated with observing the daughter momentum in the solid-angle element $\dd\Omega_{\hat r_{\mathcal A}}$ around $\hat r_{\mathcal A}$. The term $\hat r_{\mathcal A}\cdot\bm{\sigma}$ expresses that the daughter direction acts as a measurement pointer variable for the parent spin along $\hat r_{\mathcal A}$, with efficiency set by $\eta_{\mathcal A}$. Bob's wing is described analogously.

In the idealized spin-analyzer description, the detector regions selected by a setting are represented as regions on the sphere of daughter directions. For Alice, the setting is an analysis axis $\hat x$ chosen independently of the decay process; Bob's setting is denoted by $\hat y$. Each axis specifies two antipodal accepted regions, associated with the outputs $+1$ and $-1$, while the complement corresponds to the no-click output $\varnothing$. On each run, the axes are selected independently of the decay process before the observed directions are assigned to outcomes.

For Alice's setting $\hat x$, the acceptance regions are two spherical caps centered around $\hat x$ and $-\hat x$, with angular radius $\varphi$; see Fig.~\ref{fig:caps3d}. We take $0<\varphi\leq\pi/2$, so that the two caps are disjoint up to their common boundary. We define the two caps as
\begin{align}
  C_\varphi(+\hat x)
  &\coloneqq
  \{\hat r_{\mathcal A}:\hat r_{\mathcal A}\cdot\hat x\geq \cos\varphi\}, \\
  C_\varphi(-\hat x)
  &\coloneqq
  \{\hat r_{\mathcal A}:\hat r_{\mathcal A}\cdot\hat x\leq -\cos\varphi\}.
\end{align}
In the figure, $\hat r_{\mathcal A}$ is shown at the boundary of the $+1$ cap, so that $\varphi$ is the angle between the observed daughter direction and the analysis axis. The three local outcomes are
\begin{equation}
  a=
  \begin{cases}
    \pm 1, & \hat r_{\mathcal A}\in C_\varphi(\pm\hat x), \\
    \varnothing, & \text{otherwise.}
  \end{cases}
  \label{eq:local-outcomes}
\end{equation}
The corresponding effects~\cite{KrausEffects} are
\begin{align}
  E_{\pm|\hat x}^{(\varphi)}
  &=
  \int_{C_\varphi(\pm\hat x)} \dd\Omega_{\hat r_{\mathcal A}}\,
  K_{\mathcal A}(\hat r_{\mathcal A}) \\
  &=
  \frac{1-\cos\varphi}{2}\,\Id
  \pm
  \frac{\eta_{\mathcal A}(1-\cos^2\varphi)}{4}\,
  \hat x\cdot\bm{\sigma}.
  \label{eq:Epm3D}
\end{align}
Together with the no-click effect
\begin{equation}
  E_{\varnothing|\hat x}^{(\varphi)}
  =
  \Id - E_{+|\hat x}^{(\varphi)} - E_{-|\hat x}^{(\varphi)}
  =
  \cos\varphi\,\Id,
\end{equation}
these effects define a proper three-outcome POVM\@. We denote by $D_{\mathcal A}$ the local acceptance event, namely the event in which the daughter momentum falls inside one of the two accepted caps. The corresponding acceptance effect is
\begin{equation}
  E_{D_{\mathcal A}|\hat x}^{(\varphi)}
  \coloneqq
  E_{+|\hat x}^{(\varphi)}
  +
  E_{-|\hat x}^{(\varphi)}
  =
  q_{\mathcal A}\,\Id,
  \label{eq:acceptance-effect}
\end{equation}
where $q_{\mathcal A}\coloneqq 1-\cos\varphi$ gives the local acceptance probability. Thus, within the model, this probability is independent of the source state and of the chosen analysis axis. A similar construction holds for Bob. We allow, in principle, different cap radii on the two wings, but take $\varphi=\varphi_{\mathcal A}=\varphi_{\mathcal B}$ throughout for simplicity.

\textit{No-go theorem.} Here we formalize the fact that, without postselection, data comprising detections over the entire detector cannot violate a Bell inequality in a Bell test. The theorem below is stated in full generality, and then specialized to our model.

Let $\Xi_{\mathcal A}$ and $\Xi_{\mathcal B}$ denote the local outcome spaces, with outcomes $\xi_{\mathcal A}\in\Xi_{\mathcal A}$ and $\xi_{\mathcal B}\in\Xi_{\mathcal B}$. We write their joint probability measure as $\nu(\dd\xi_{\mathcal A},\dd\xi_{\mathcal B})$, where $\dd\xi_{\mathcal A}$ and $\dd\xi_{\mathcal B}$ denote the corresponding measure elements. A classical post-processing maps these outcomes, possibly stochastically, to Bell outputs. It is local and setting-dependent if Alice's map uses only $(\xi_{\mathcal A},x)$ and Bob's only $(\xi_{\mathcal B},y)$, and complete if each map assigns total probability one to the Bell outputs for every underlying outcome.

\begin{theorem}
  \label{thm:no-go-postprocessing}
  If the joint measure $\nu$ is independent of the settings $x$ and $y$, then any correlations obtained by a complete local, setting-dependent classical post-processing admit a local hidden-variable model and hence are Bell local.
\end{theorem}

\begin{proof}
  No factorization of $\nu$ is assumed, so the underlying outcomes may be arbitrarily correlated. Represent Alice's post-processing by weights $w^{\mathcal A}_{a|x}(\xi_{\mathcal A})$, giving the probability of recording $a$ conditional on $x$ and $\xi_{\mathcal A}$, with
  \begin{equation}
    w^{\mathcal A}_{a|x}(\xi_{\mathcal A})\geq0,
    \qquad
    \sum_a w^{\mathcal A}_{a|x}(\xi_{\mathcal A})=1,
  \end{equation}
  and analogously for Bob. The locality of the post-processing then gives
  \begin{equation}
    p(a,b|x,y)=
    \int_{\Xi_{\mathcal A}\times\Xi_{\mathcal B}}
    \nu(\dd\xi_{\mathcal A},\dd\xi_{\mathcal B})\,
    w^{\mathcal A}_{a|x}(\xi_{\mathcal A})\,
    w^{\mathcal B}_{b|y}(\xi_{\mathcal B}).
    \label{eq:coarse-grained-probs}
  \end{equation}
  This is already a local hidden-variable decomposition. Indeed, take $\lambda=(\xi_{\mathcal A},\xi_{\mathcal B})$, with probability measure
  $\rho(\dd\lambda)\coloneqq\nu(\dd\xi_{\mathcal A},\dd\xi_{\mathcal B})$, and define
  $p(a|x,\lambda)=w^{\mathcal A}_{a|x}(\xi_{\mathcal A})$ and
  $p(b|y,\lambda)=w^{\mathcal B}_{b|y}(\xi_{\mathcal B})$.
  Equation~\eqref{eq:coarse-grained-probs} then has the Bell-local form of Eq.~\eqref{eq:bell-lhv-decomp}, and hence all Bell inequalities are satisfied.
\end{proof}

In the present model, $(\xi_{\mathcal A},\xi_{\mathcal B})=(\hat r_{\mathcal A},\hat r_{\mathcal B})$, and the POVMs of Eq.~\eqref{eq:spin-analyzer-povm} give the joint angular probability density
\begin{equation}
  p(\hat r_{\mathcal A},\hat r_{\mathcal B})
  =
  \Tr\!\left[
    \rho_{\mathcal{AB}}
    K_{\mathcal A}(\hat r_{\mathcal A})\otimes K_{\mathcal B}(\hat r_{\mathcal B})
  \right].
  \label{eq:full-angular-distribution}
\end{equation}
The post-processing weights $w^{\mathcal A}_{a|\hat x}(\hat r_{\mathcal A})$ are determined by the detector regions selected by Alice's setting, and analogously for Bob. Thus, as long as all detector outcomes are retained, we obtain the local-hidden variable decomposition
\begin{equation}
  p(a,b|\hat x,\hat y)
  =
  \int \dd\Omega_{\hat r_{\mathcal A}}\,\dd\Omega_{\hat r_{\mathcal B}}\,
  p(\hat r_{\mathcal A},\hat r_{\mathcal B})\,
  w^{\mathcal A}_{a|\hat x}(\hat r_{\mathcal A})\,
  w^{\mathcal B}_{b|\hat y}(\hat r_{\mathcal B}).
  \label{eq:model-coarse-grained-probs}
\end{equation}

\textit{Joint-acceptance contextuality.} Because complete local post-processing of the full detector record cannot yield a Bell violation, postselection becomes necessary. This, however, opens a detection or postselection loophole. Whether the postselected data can still define a sample for certifying Bell nonlocality depends on measurement contextuality assumptions.

Let $D\coloneqq D_{\mathcal A}\cap D_{\mathcal B}$ denote the joint-acceptance event. Within the spin-analyzer model, its effect is
\begin{equation}
  E_{D|\hat x,\hat y}
  =
  E_{D_{\mathcal A}|\hat x}\otimes E_{D_{\mathcal B}|\hat y}
  =
  q_{\mathcal A}q_{\mathcal B}\,\Id_{\mathcal{AB}},
  \label{eq:joint-acceptance-effect}
\end{equation}
independently of the setting pair. For a Bell-local hidden-variable model satisfying measurement independence before postselection, the accepted distribution can be written as
\begin{align}
  & p(a,b|\hat x,\hat y,D) \nonumber\\
  & \quad =
  \int \dd\lambda\,
  \rho(\lambda|\hat x,\hat y,D)\,
  p(a|\hat x,\lambda,D_{\mathcal A})\,
  p(b|\hat y,\lambda,D_{\mathcal B}),
\end{align}
where
\begin{align}
  \rho(\lambda|\hat x,\hat y,D)
  =
  \frac{
    \rho(\lambda)\,
    p(D_{\mathcal A}|\hat x,\lambda)\,
    p(D_{\mathcal B}|\hat y,\lambda)
  }{
    \int \dd\lambda'\,\rho(\lambda')\,
    p(D_{\mathcal A}|\hat x,\lambda')\,
    p(D_{\mathcal B}|\hat y,\lambda')
  }.
\end{align}
The loophole is that the accepted hidden-variable distribution may depend on the measurement settings,
\begin{equation}
  \rho(\lambda|\hat x,\hat y,D)
  \neq
  \rho(\lambda|D),
\end{equation}
so that different correlators in a Bell expression may be evaluated on different hidden-variable subensembles.

A specific local hidden-variable construction exhibiting this loophole is obtained by taking $\lambda=(\hat r_{\mathcal A},\hat r_{\mathcal B})$, following Kasday's approach~\cite{Kasday1971} as adapted to collider settings in Ref.~\cite{Abel1992}. In this model, postselection simply retains the runs in which the hidden-variable directions happen to fall inside the caps selected by the measurement settings. Thus the accepted subensemble contains a correlation between $\lambda$ and the settings, with the cap radius $\varphi$ fixing the tolerance. More specifically, for Alice, and analogously for Bob, $p(D_{\mathcal A}|\hat x,\hat r_{\mathcal A})=1$ if $\hat r_{\mathcal A} \in C_\varphi(+\hat x) \cup C_\varphi(-\hat x)$ and $0$ otherwise.

The spin-analyzer model, however, predicts that the acceptance probabilities are independent of both the settings and the preparation; see Eqs.~\eqref{eq:acceptance-effect} and~\eqref{eq:joint-acceptance-effect}. This motivates treating the local and joint-acceptance events as operationally equivalent across settings, a claim that must ultimately be tested experimentally. In the local construction above, by contrast, this operationally invisible distinction reappears as a different filtering of $\lambda$.

Here is the relevant form of measurement contextuality. For Alice, let $P_{\mathcal A}$ denote an arbitrary preparation, including, for example, a conditional preparation defined by a measurement outcome on Bob's wing. If the local acceptance events are operationally equivalent, then
\begin{equation}
  p(D_{\mathcal A}|\hat x, P_{\mathcal A})
  =
  p(D_{\mathcal A}|\hat x', P_{\mathcal A})
  \qquad \forall P_{\mathcal A}.
\end{equation}
Such operational equivalences are expected, for example, in decays of spin-$\tfrac{1}{2}$ particle pairs produced in an overall scalar or pseudoscalar state. Acceptance noncontextuality then requires this operational equivalence to be reflected at the hidden-variable level:
\begin{equation}
  p(D_{\mathcal A}|\hat x,\lambda)
  =
  p(D_{\mathcal A}|\hat x',\lambda)
  \qquad \forall \lambda.
\end{equation}
The analogous condition holds for Bob. If, moreover, the acceptance event is identified with the trivial procedure that accepts with probability $q_{\mathcal A}$ independently of the system, and similarly for Bob, then noncontextuality gives
\begin{equation}
  p(D_{\mathcal A}|\hat x,\lambda)=q_{\mathcal A},
  \qquad
  p(D_{\mathcal B}|\hat y,\lambda)=q_{\mathcal B},
  \qquad \forall \lambda.
  \label{eq:strong-fair-sampling}
\end{equation}
This is the strong condition for fair sampling, under which
\begin{equation}
  \rho(\lambda|\hat x,\hat y,D)=\rho(\lambda),
\end{equation}
and the postselected subensemble can be treated as a proper Bell-test sample.

The strong fair-sampling condition is sufficient for a valid Bell test on the postselected data, but it is stronger than needed. The accepted hidden-variable subensemble may differ from the original ensemble, provided it is independent of the chosen setting pair. This leads to the following weaker condition.

\begin{theorem}
  \label{thm:joint-acceptance-noncontextuality}
  Assume measurement independence before postselection, and let the joint-acceptance events be operationally equivalent for all setting pairs. If the underlying local hidden-variable theory is joint-acceptance noncontextual, then
  \begin{equation}
    \rho(\lambda|\hat x,\hat y,D)
    =
    \rho(\lambda|D).
    \label{eq:fair-sampling-postselection}
  \end{equation}
\end{theorem}

\begin{proof}
  By joint-acceptance noncontextuality, operational equivalence of the joint-acceptance events implies
  \begin{equation}
    p(D|\hat x,\hat y,\lambda)=p(D|\lambda)
  \end{equation}
  for all $\hat x$ and $\hat y$, where $p(D|\lambda)$ denotes their common response function. Using Bayes' theorem and measurement independence before postselection, we then obtain
  \begin{align}
    \rho(\lambda|\hat x,\hat y,D)
    &=
    \frac{
      \rho(\lambda|\hat x,\hat y)\,
      p(D|\hat x,\hat y,\lambda)
    }{
      \int \dd\lambda'\,
      \rho(\lambda'|\hat x,\hat y)\,
      p(D|\hat x,\hat y,\lambda')
    } \\
    &=
    \frac{
      \rho(\lambda)\,p(D|\lambda)
    }{
      \int \dd\lambda'\,\rho(\lambda')\,p(D|\lambda')
    }
    =
    \rho(\lambda|D).
  \end{align}
  Hence Eq.~\eqref{eq:fair-sampling-postselection} follows.
\end{proof}

\textit{Bell inequality.} Equation~\eqref{eq:fair-sampling-postselection} ensures that postselection does not introduce correlations between $\lambda$ and the settings. The postselected data therefore define an effective Bell scenario with continuous inputs $\hat x,\hat y\in S^2$ and binary outputs $a,b\in\{\pm1\}$. We now introduce a full-correlation Bell inequality adapted to the rotational geometry of the spin-analyzer model.

\begin{theorem}
  \label{thm:continuous-input-bell-inequality}
  For independently and uniformly sampled settings $\hat x,\hat y\in S^2$ and binary outputs $a,b\in\{\pm1\}$, let $\E$ denote the expectation over both the settings and outcomes, and define
  \begin{align}
    \mathcal B
    &\coloneqq
    -\E\!\left[
      \hat x\cdot\hat y\,ab
    \right] \nonumber \\
    &=
    -
    \int \frac{\dd\Omega_x}{4\pi}
    \int \frac{\dd\Omega_y}{4\pi}\,
    \hat x\cdot\hat y
    \sum_{a,b} a\,b\,p(a,b|\hat x,\hat y).
    \label{eq:bell-functional}
  \end{align}
  Then, for any measurement-independent local hidden-variable model,
  \begin{equation}
    \mathcal B\leq \frac14.
    \label{eq:classical-bell-bound}
  \end{equation}
\end{theorem}

\begin{proof}
  Let
  \begin{equation}
    \E\left[a_{\hat x}\right]_\lambda
    \coloneqq
    \sum_{a=\pm1} a\,p(a|\hat x,\lambda),
  \end{equation}
  and likewise for $\E\left[b_{\hat y}\right]_\lambda$. From the local hidden-variable decomposition, the correlator can be expressed as
  \begin{equation}
    \E\left[a_{\hat x}\,b_{\hat y}\right]
    =
    \int \dd\lambda\,\rho(\lambda)\,
    \E\left[a_{\hat x}\right]_\lambda
    \E\left[b_{\hat y}\right]_\lambda.
  \end{equation}
  Defining
  \begin{equation}
    \vec\alpha_\lambda
    \coloneqq
    \int \frac{\dd\Omega_x}{4\pi}
    \E\left[a_{\hat x}\right]_\lambda\hat x,
    \qquad
    \vec\beta_\lambda
    \coloneqq
    \int \frac{\dd\Omega_y}{4\pi}
    \E\left[b_{\hat y}\right]_\lambda\hat y,
  \end{equation}
  one can write
  \begin{equation}
    \mathcal B
    =
    -\int \dd\lambda\,\rho(\lambda)\,\vec\alpha_\lambda\cdot\vec\beta_\lambda.
  \end{equation}
  Since $\E\left[a_{\hat x}\right]_\lambda,\E\left[b_{\hat y}\right]_\lambda\in[-1,1]$, for any unit vector $\hat u$,
  \begin{align}
    \hat u\cdot\vec\alpha_\lambda
    &=
    \int \frac{\dd\Omega_x}{4\pi}
    \E\left[a_{\hat x}\right]_\lambda\,\hat u\cdot\hat x \leq
    \int \frac{\dd\Omega_x}{4\pi}
    |\hat u\cdot\hat x|
    =
    \frac12,
  \end{align}
  thus $|\vec\alpha_\lambda|\leq1/2$, and analogously $|\vec\beta_\lambda|\leq1/2$. Hence
  \begin{equation}
    -\vec\alpha_\lambda\cdot\vec\beta_\lambda
    \leq
    |\vec\alpha_\lambda|\,|\vec\beta_\lambda|
    \leq
    \frac14.
  \end{equation}
  Averaging over $\lambda$ preserves the bound, so $\mathcal B\leq1/4$.
\end{proof}

Quantum mechanics predicts that this bound can be violated. To see this, we consider the renormalized detection effects obtained by conditioning on the postselected events. Using Eqs.~\eqref{eq:Epm3D} and~\eqref{eq:acceptance-effect}, one obtains
\begin{equation}
  \widetilde E_{\pm|\hat x}^{(\varphi)}
  =
  \frac{E_{\pm|\hat x}^{(\varphi)}}{q_{\mathcal A}}
  =
  \frac12
  \left(
    \Id
    \pm
    \mu_{\mathcal A}\,\hat x\cdot\bm{\sigma}
  \right),
\end{equation}
with $\mu_{\mathcal A}=\eta_{\mathcal A}(1+\cos\varphi)/2$,
and analogously for Bob. The corresponding binary observables are therefore $\widetilde a_{\hat x}=\mu_{\mathcal A}\,\hat x\cdot\bm{\sigma}$ and $\widetilde b_{\hat y}=\mu_{\mathcal B}\,\hat y\cdot\bm{\sigma}$.
For a singlet source state, this gives
\begin{equation}
  \E\left[\widetilde a_{\hat x}\widetilde b_{\hat y}\right]
  =
  -\mu_{\mathcal A}\mu_{\mathcal B}\,\hat x\cdot\hat y.
\end{equation}
Substituting this into the Bell functional yields
\begin{align}
  \mathcal B_{\mathrm{QM}}
  &=
  \mu_{\mathcal A}\mu_{\mathcal B}
  \int \frac{\dd\Omega_x}{4\pi}
  \int \frac{\dd\Omega_y}{4\pi}
  (\hat x\cdot\hat y)^2 =
  \frac{\mu_{\mathcal A}\mu_{\mathcal B}}{3}.
  \label{eq:quantum-bell-value}
\end{align}
Thus the classical bound is violated whenever
\begin{equation}
  \mu_{\mathcal A}\mu_{\mathcal B}>\frac34.
  \label{eq:quantum-violation-threshold}
\end{equation}
In the symmetric case, $\mu_{\mathcal A}=\mu_{\mathcal B}=\mu$, this becomes $\mu>\sqrt3/2$.
In the ideal sharp limit, $\eta_{\mathcal A}=\eta_{\mathcal B}=1$ and $\varphi\to0$, one has $\mu_{\mathcal A}=\mu_{\mathcal B}=1$, so that $\mathcal B_{\mathrm{QM}}=1/3>1/4$. Accordingly, our main result is that, under measurement independence and operational equivalence of the joint-acceptance events, a violation of this bound witnesses Bell nonlocality or joint-acceptance contextuality.

\textit{Conclusion.} We have formulated operational Bell tests for collider decay processes in which independently selected analysis axes provide the inputs and coarse-grained detector regions define the outputs, as illustrated in Fig.~\ref{fig:caps3d}. We first showed in Theorem~\ref{thm:no-go-postprocessing} that, when the settings determine only complete local post-processings of a setting-independent detector record, the resulting correlations necessarily admit a local hidden-variable model. Thus, if the settings are introduced only through the analysis of a fixed detector record, a Bell violation requires postselection, implemented here by retaining only events in narrow antipodal regions.

Postselection may prevent such a violation from certifying Bell nonlocality. Theorem~\ref{thm:joint-acceptance-noncontextuality} shows, however, that operational equivalence of the acceptance events prevents this loophole from being exploited by any noncontextual local model: the accepted hidden-variable distribution cannot acquire a dependence on the chosen settings. We then derived the continuous-input Bell inequality of Theorem~\ref{thm:continuous-input-bell-inequality}, $\mathcal B\leq1/4$, whose quantum violation occurs whenever the condition in Eq.~\eqref{eq:quantum-violation-threshold} is satisfied. Subject to measurement independence and the relevant operational equivalences, a violation therefore witnesses Bell nonlocality or joint-acceptance contextuality. Although the explicit construction concerns spin-$\tfrac{1}{2}$ particles, the same operational strategy may admit extensions to higher-spin systems with suitably adapted coarse grainings and Bell inequalities. Since both the operational equivalences and the Bell functional can be tested using detector-level statistics, this construction provides a route from collider state-reconstruction studies to operational tests of classicality in high-energy processes.

\begin{acknowledgments}
This work was made possible through the support of Grant 63206 from the John Templeton Foundation. The work of AJB is also funded through STFC grants ST/R002444/1 and ST/S000933/1, the Binks Trust, and the John Fell Oxford University Press Research Fund. ACO acknowledges that this study was financed, in part, by the São Paulo Research Foundation (FAPESP), Brazil, process numbers 2025/10927-1 and 2026/08210-4, and by the QuantERA II Programme, which received funding from the European Union's Horizon 2020 research and innovation programme under Grant Agreement No.~101017733, project ``PhoMemtor,'' No.~2021/03/Y/ST2/00177. The authors are grateful to the members of the Oxford Quantum Collider research group for helpful comments and suggestions, to Jonathan Barrett for helpful discussions, and to the organisers of the ``Quantum Observables for Collider Physics 2026'' workshop at CERN where these proposals were first presented. The opinions expressed in this publication are those of the authors and do not necessarily reflect the views of the funding bodies. This work is distributed under the terms of the Creative Commons Attribution 4.0 International Licence (CC BY 4.0), which permits unrestricted use, distribution, and reproduction in any medium, provided the original author and source are credited.
\end{acknowledgments}

\bibliography{references}

\section*{End Matter}

\appendix*

We outline an experimental procedure for implementing the measurements considered in the main text, testing the required operational equivalences, and evaluating the Bell inequality. The equivalence tests and all optimization steps must be performed on datasets disjoint from the sample used for the Bell test itself.

\section{\label{app:physical-setup}Physical setup}

A suitable process should produce an entangled pair of unstable particles whose decay products have high spin-analyzing power. The directions $\hat r_{\mathcal A}$ and $\hat r_{\mathcal B}$ should also be reliably inferable from detector-level momenta. In general, these directions are defined in the parent rest frames, whose reconstruction may depend on assumptions about the decay and background effects. Such dependence can be reduced by selecting events in which the parent boosts are sufficiently small and well constrained, so that the measured laboratory-frame directions closely track the relevant rest-frame directions. The resulting reconstruction uncertainty must be included in the experimental response model. The impact of detector efficiency and acceptance effects can be characterized using independent control samples from other physics processes. Corrections inferred from these samples may then be applied, provided that the selected control events give an unbiased estimate of the relevant detector response. Top quark pairs are natural candidates because they are produced in large numbers and the empirically successful quantum model of their decays predicts spin-analyzing power close to unity. This prediction makes a violation of our Bell inequality likely, but no assumed value of spin-analyzing power enters the Bell-test argument itself.

In dileptonic top--antitop decays, the two wings can be identified by the opposite charges of the detected leptons. Each wing is supplied with an independent, time-stamped stream of unit vectors, uniformly distributed over the sphere. A stream may be obtained by converting independently sampled random bits into spherical coordinates through a predetermined map. Its generation and processing must be independent of the source preparation, event-selection criteria, and stream used on the opposite wing. Cosmic sources provide one possible implementation, for which any common cause correlating the settings with the preparation is pushed into the distant past.

\section{\label{app:operational-equivalence}Operational equivalence}

Assuming measurement independence and noncontextuality, Theorem~\ref{thm:joint-acceptance-noncontextuality} requires only operational equivalence of joint-acceptance events across setting pairs to obtain Eq.~\eqref{eq:fair-sampling-postselection}. Here, however, we propose testing the stronger operational equivalences of the local acceptance events, which imply the strong fair-sampling condition in Eq.~\eqref{eq:strong-fair-sampling}.

Ideally, this requires establishing
\begin{align}
  p(D_{\mathcal A}|\hat x,P_{\mathcal A})
  &=q_{\mathcal A}
  &&\forall\,\hat x,P_{\mathcal A},
  \label{eq:local-acceptance-operational-equivalence-alice} \\
  p(D_{\mathcal B}|\hat y,P_{\mathcal B})
  &=q_{\mathcal B}
  &&\forall\,\hat y,P_{\mathcal B}.
  \label{eq:local-acceptance-operational-equivalence-bob}
\end{align}
A finite, model-independent experiment cannot assume in advance that a given preparation set is tomographically complete, exhaustively cover a continuous range of analysis axes, or establish exact equality of probabilities. We therefore consider a redundant family of preparations, collect statistics for a large set of measurement settings, infer the rank of the resulting response tables, and validate the inferred preparation space using predictive tests.

We begin by selecting three mutually orthogonal choices of $\hat x$ on Alice's wing and conditioning separately on the outcomes $\{+,-\}$, thereby defining six preparations on Bob's wing. Bob characterizes these preparations using the full probabilities of the outcomes $b\in\{+,-,\varnothing\}$ for many different axes $\hat y$. Since, without assuming quantum theory, the orthogonal triad cannot be assumed to provide tomographic completeness, we augment this set with preparations conditioned on Alice's outcome $\varnothing$ and on the outcomes of additional, randomly chosen axes.

The collected data define the response table $R^{\mathcal B}_{i,(j,b)} = p(b|\hat y_j, P^{\mathcal B}_i)$. The affine dimension of the normalized preparation space resolved by these measurements is $d_{\mathcal B} = \text{rank}\,R^{\mathcal B}_{i,(j,b)} - 1$. In practice, the rank must be estimated while accounting for statistical fluctuations, which generally make the observed frequency table full rank. Part of the data is used to determine the dimension and identify a set of preparations spanning the observed preparation space. Additional preparations and analysis axes are then used to test whether their statistics can be predicted within this space. Successful predictions support the assumption that the experimentally accessible preparation space has been adequately characterized.

The procedure is then repeated with Alice and Bob interchanged. Once both preparation spaces have been characterized, the acceptance probabilities are tested across the preparations and analysis axes to determine whether they take the common values $q_{\mathcal A}$ and $q_{\mathcal B}$ required by Eqs.~\eqref{eq:local-acceptance-operational-equivalence-alice} and~\eqref{eq:local-acceptance-operational-equivalence-bob}. Because the experiment involves finite statistics and only finitely many preparations and axes, these tests provide operational evidence for the equivalences rather than establishing them exactly.

\section{\label{app:bell-test}Bell test}

Upon obtaining robust evidence for the operational equivalence of the local acceptance events, the event-selection criteria, analysis parameters, and cap width used in the equivalence tests are fixed and carried over unchanged to the Bell test. Using only calibration or simulation data, the cap width should be chosen to maximize the expected statistical significance, balancing the stronger correlations obtained with narrower caps against their lower acceptance rate, as explained below. If the operational-equivalence tests fail, the procedure must be revised and revalidated before examining the Bell sample.

Each detection event is matched by its timestamp to the corresponding analysis axes $\hat x$ and $\hat y$ for Alice and Bob, respectively. Events for which $a=\varnothing$ or $b=\varnothing$ are discarded from the Bell-test sample, leaving only those for which the joint-acceptance event $D=D_{\mathcal A}\cap D_{\mathcal B}$ occurs. We call these jointly accepted trials. Denoting their number by $N_D$, we estimate the continuous Bell functional $\mathcal B$ defined in Eq.~\eqref{eq:bell-functional} using the finite-sample estimator
\begin{equation}
  \widehat{\mathcal B}
  =
  -\frac{1}{N_D} \sum_{i=1}^{N_D} (\hat{\boldsymbol x}_i\cdot\hat{\boldsymbol y}_i)a_i b_i.
  \label{eq:bell-functional-estimator}
\end{equation}

Lastly, we outline how to choose the cap radius $\phi$ that maximizes the expected statistical significance of a Bell-inequality violation. The predicted separation between the quantum value in Eq.~\eqref{eq:quantum-bell-value} and the classical bound in Eq.~\eqref{eq:classical-bell-bound} is
\begin{equation}
  \Delta(\phi)
  :=
  \frac{\mu_{\mathcal A}(\phi)\mu_{\mathcal B}(\phi)}{3}
  -\frac{1}{4}.
\end{equation}
For independent trials, $a^2b^2=1$ and $\mathbb{E}[(\hat{\boldsymbol x}\cdot\hat{\boldsymbol y})^2]=1/3$, so the standard error of the finite-sample estimator is
\begin{equation}
  \sigma_{\widehat{\mathcal B}}(\phi)
  =
  \frac{1}{\sqrt{N_D}}
  \sqrt{
    \frac{1}{3}
    -
    \frac{
      [\mu_{\mathcal A}(\phi)\mu_{\mathcal B}(\phi)]^2
    }{9}
  }.
\end{equation}
Writing $N_0$ for the number of detection events before cap postselection, the expected number of jointly accepted trials is $N_D\simeq N_0(1-\cos\phi)^2$. The expected statistical significance can be expressed as
\begin{equation}
  S(\phi)
  =
  \sqrt{N_0}(1-\cos\phi)
  \frac{
    \dfrac{\mu_{\mathcal A}(\phi)\mu_{\mathcal B}(\phi)}{3}
    -\dfrac{1}{4}
  }{
    \sqrt{
      \dfrac{1}{3}
      -
      \dfrac{
        [\mu_{\mathcal A}(\phi)\mu_{\mathcal B}(\phi)]^2
      }{9}
    }
  },
  \label{eq:expected-statistical-significance}
\end{equation}
where $\mu_{\mathcal A}(\phi)=\eta_{\mathcal A}(1+\cos\phi)/2$, and analogously for Bob. Given the spin-analyzing powers, the optimal cap radius is obtained numerically from
\begin{equation}
  \phi_{\mathrm{opt}}
  =
  \operatorname*{arg\,max}_{\phi} S(\phi),
\end{equation}
where the maximization is restricted to experimentally admissible values of $\phi$.

\end{document}